\documentclass{llncs}
\usepackage{todonotes}
\usepackage{graphicx}
\usepackage[inline]{enumitem}
\usepackage{lineno}
\usepackage{xspace}
\usepackage{amsmath}
\usepackage{amssymb}
\usepackage{hyperref}
\usepackage{cleveref}
\usepackage{subcaption}
\usepackage{thm-restate}
\usepackage{wrapfig}

\crefname{property}{Property}{Properties}

\let\oldendproof\endproof
\renewcommand\endproof{~\hfill$\qed$\oldendproof}

\newcommand{\UER}{UER-RF\xspace}
\newcommand{\UERII}{UER-USF\xspace}
\newcommand{\UERlong}{Unit Edge-length Rectilinear drawings with Rectangular Faces\xspace}
\newcommand{\UERIIlong}{Unit Edge-length Rectilinear drawings with Unit Square Faces\xspace}
\newcommand{\leftA}{\textsf{Left}\xspace}
\newcommand{\upA}{\textsf{Up}\xspace}
\newcommand{\nullA}{\textsf{null}\xspace}

\title{Unit Edge-Length Rectilinear Drawings with Crossings and Rectangular Faces\thanks{This research was supported, in part,  by MUR of Italy (PRIN Project no.~2022ME9Z78~-- NextGRAAL and PRIN Project no.~2022TS4Y3N~-- EXPAND). The sixth author was supported by Ce.Di.Pa. - PNC Programma unitario di interventi per le aree  del terremoto del 2009-2016 - Linea di intervento 1 sub-misura B4 - "Centri di ricerca per l'innovazione" CUP J37G22000140001.}} 

\author{Patrizio Angelini\inst{1}
\and
Carla Binucci\inst{2}
\and
Giuseppe Di Battista\inst{3}
\and
Emilio Di Giacomo\inst{2}
\and
Walter Didimo\inst{2}
\and
Fabrizio Grosso\inst{3}
\and
Giacomo Ortali\inst{2}
\and
Ioannis G. Tollis\inst{4}
}

\institute{John Cabot University \email{pangelini@johncabot.edu} \and University of Perugia 
\email{\textit{firstname.lastname}@unipg.it}
\and  University of Roma Tre 
\email{\textit{firstname.lastname}@uniroma3.it}
\and University of Crete \email{tollis@csd.uoc.gr}}

\begin{document}

\maketitle



\begin{abstract}
Unit edge-length drawings, rectilinear drawings (where each edge is either a horizontal or a vertical segment), and rectangular face drawings are among the most studied subjects in Graph Drawing. 
However, most of the literature on these topics refers to planar graphs and planar drawings.
In this paper we study drawings with all the above nice properties but that can have edge crossings; we call them \UERlong (\UER drawings). We consider crossings as dummy vertices and apply the unit edge-length convention to the edge segments connecting any two (real or dummy) vertices. 
%
Note that \UER drawings are grid drawings (vertices are placed at distinct integer coordinates), which is another classical requirement of graph visualizations.
%
%
We present several efficient and easily implementable algorithms for recognizing graphs that admit \UER drawings and for constructing such drawings if they exist. We consider restrictions on the degree of the vertices or on the size of the faces. For each type of restriction, we consider both the general unconstrained setting and a setting in which either the external boundary of the drawing is fixed or the rotation system of the graph is fixed as part of the input.    
%
\end{abstract}

\section{Introduction}\label{se:intro}

Planar graph drawings where all edges have unit length and are represented as either horizontal or vertical segments (rectilinear drawings), and where all faces are convex (i.e., rectangles), adhere to several of the most celebrated Graph Drawing aesthetics. Further,~such drawings can be translated in the plane so that their vertices have integer coordinates. Motivated by these considerations, a recent paper investigated these types of planar layouts~\cite{Alegria22}.

We extend the above drawing convention to non-planar drawings, considering each crossing as a ``dummy'' vertex and applying the unit edge-length constraint to the edge segments connecting any two vertices, either real or dummy. In this setting, we study \emph{\UERlong} (\emph{\UER drawings}), i.e., rectilinear drawings in which all the edges have unit length and all the faces (including the external one) are rectangles. Examples of \UER drawings are presented in \Cref{fig:crd-examples}. Note that, in addition to guaranteeing edges without bends and rectangular faces, \UER drawings (when they exist) are typically more compact than orthogonal grid drawings computed with classical algorithms, which aim to minimize crossings and bends at the same time; for a visual comparison, see for example \Cref{fig:gdt}.  

\begin{figure}[tb]
    \centering
    \begin{subfigure}{0.25\textwidth}
            \includegraphics[width=\textwidth,page=3]{./figures/crd-examples.pdf}
            \subcaption{$5 \times 5$}
            \label{fig:crd-examples-a}
    \end{subfigure}
    \hfil
    \begin{subfigure}{0.25\textwidth}
            \includegraphics[width=\textwidth,page=2]{./figures/crd-examples.pdf}
            \subcaption{$5 \times 5$}
            \label{fig:crd-examples-b}
        \end{subfigure}
        \hfil
        \begin{subfigure}{0.25\textwidth}
            \includegraphics[width=\textwidth,page=1]{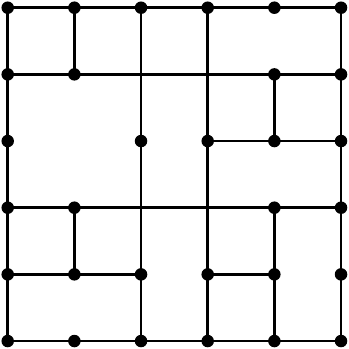}
            \subcaption{$5 \times 5$}
            \label{fig:crd-examples-c}
        \end{subfigure}
    \caption{\UER drawings: with Unit Square Faces (a), with no internal degree-3 vertex (b), and in the general setting (c). All drawings occupy a $5 \times 5$ grid area.}
    \label{fig:crd-examples}
\end{figure}

\begin{figure}[h]
    \centering
    \begin{subfigure}{0.32\textwidth}
            \includegraphics[width=\textwidth,page=3]{./figures/gdt.pdf}
            \subcaption{$11 \times 10$}
            \label{fig:gdt-a}
    \end{subfigure}
    \hfil
    \begin{subfigure}{0.32\textwidth}
            \includegraphics[width=\textwidth,page=2]{./figures/gdt.pdf}
            \subcaption{$11 \times 11$}
            \label{fig:gdt-b}
        \end{subfigure}
        \hfil
        \begin{subfigure}{0.32\textwidth}
            \includegraphics[width=\textwidth,page=1]{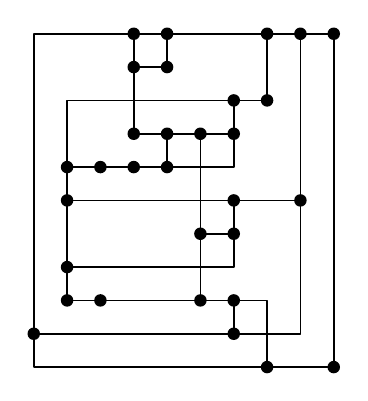}
            \subcaption{$10 \times 11$}
            \label{fig:gdt-c}
        \end{subfigure}
    \caption{The same graphs depicted in \Cref{fig:crd-examples} drawn with the topology-shape-metrics algorithm implemented in the GDToolkit library \cite{DBLP:reference/crc/BattistaD13}. Below each drawing we report the width and the height in terms of grid units. These dimensions are at least twice those of the drawings in \Cref{fig:crd-examples}; the three drawings have been scaled-down to fit in the same row.}
    \label{fig:gdt}
\end{figure}

The problem we study is related to several challenging problems.
Recognizing graphs that admit planar straight-line drawings with all edges of the same length is NP-hard~\cite{cdr-pegse-07,DBLP:journals/dam/EadesW90} and, even stronger, $\exists\mathbb{R}$-complete~\cite{DBLP:conf/compgeom/AbelDDELS16,Schaefer2013}.
Testing whether a graph is rectilinear planar is NP-hard~\cite{DBLP:journals/siamcomp/GargT01}, as well as, recognizing graphs that admit unit edge-length drawings with integer coordinates is NP-complete even for trees~\cite{DBLP:journals/ipl/BhattC87,g-ulebt-89,gsz-grcpc-21}.
On the other hand, recognizing planar graphs that admit rectilinear drawings with rectangular faces (with unconstrained edge lengths) can be done in polynomial time if the planar embedding is part of the input~\cite{DBLP:journals/jal/RahmanNG04,DBLP:journals/siamcomp/Tamassia87}; see also~\cite{f-rsrpg-13,DBLP:books/ws/NishizekiR04,nr-rda-13}.
Additionally, observe that \UER drawings are RAC drawings (for surveys see \cite{DBLP:conf/esa/AngeliniBK00U23,DBLP:books/sp/20/Didimo20}) and that non-planar drawings are intensively studied in the beyond planarity area (for surveys see \cite{DBLP:journals/csur/DidimoLM19,DBLP:books/sp/20/Hong20}).

Alegr{\'{\i}}a et al.~\cite{Alegria22} show that recognizing planar graphs admitting planar \UER drawings is feasible in polynomial time. We aim to extend this result to non-planar graphs. We study the recognition and the construction of \UER drawings under various conditions: when the rotation system of the vertices is part of the input, when it is not, and when the external face is part of the input or not.
In our problem, an important role is played by the assignment of angles around the vertices. Requiring rectangular faces enforces $180^\circ$ angles at degree-2 vertices, and the difficulty emerges in presence of degree-3 vertices. In view of this, we present incremental restricted scenarios in which we are able to solve the problem efficiently, and then give a fixed-parameter-tractable (FPT) algorithm for the general case, parameterized by the number of degree-3 vertices.
Additionally, we consider \UER drawings in which all faces, except the external one, are unit-area squares. We refer to these drawings as \emph{\UERIIlong} (\emph{\UERII drawings}); they have the nice feature of fully exploiting the available rows and columns; see \Cref{fig:crd-examples-a} for an example.
Let $G$ be an $n$-vertex graph; we present the following~results:

\medskip\noindent{$-$} An $O(n)$-time algorithm that tests whether $G$ admits a \UERII drawing. The algorithm can be adapted to preserve a prescribed external face (\Cref{se:uerii}).

\medskip\noindent{$-$} Polynomial-time testing algorithms for \UER drawings in restricted scenarios (\Cref{se:restricted-scenarios}). We consider drawings without internal degree-3 vertices, and drawings such that removing the external face yields a collection of paths and cycles. The complexity of these algorithms ranges from $O(n)$ to $O(n^5)$, depending on the specific constraints on the external face and on the rotation system.    

\medskip\noindent{$-$} An $O(3^k n^{4.5})$-time testing algorithm for general \UER drawings, where $k$ is the number of degree-3 vertices (\Cref{se:uerr}).

%

\medskip All our algorithms are easily implementable and can preserve a given rotation system for the vertices, if needed. If the recognition is successful, within the same time complexity of the test, they can construct the corresponding drawing.



\section{Basic Definitions}\label{se:preliminaries}

Let $V(G)$ and $E(G)$ denote the set of vertices and the set of edges of a graph~$G$, respectively. Also, for a vertex $v \in V(G)$, 
let $N(v)$ be the set of \emph{neighbors of $v$}, i.e., the set of vertices adjacent to~$v$. The value $|N(v)|$ is the \emph{degree of $v$}, and is denoted by $\deg_G(v)$. For a positive integer $k$, a graph $G$ is a \emph{$k$-graph} if $\deg_G(v) \leq k$ for each $v \in V(G)$.  
If $\deg_G(v)=2$ and $N(v) = \{u,w\}$, \emph{smoothing} $v$ consists of removing $v$ (and its incident edges) from $G$ and adding the edge $(u,w)$ in $G$ (i.e., smoothing a vertex is the reverse of subdividing an edge).

\medskip
\noindent \textbf{Rectilinear Drawings.} 
A \emph{rectilinear drawing} $\Gamma$ of a graph $G$ is a drawing such that: $(i)$ each vertex $v \in V(G)$ is mapped to a distinct point $p_v$ of an integer grid; $(ii)$ each edge $e=(u,v) \in E(G)$ is mapped to either a horizontal or a vertical segment $s_e$ connecting~$p_u$ and~$p_v$; $(iii)$ if $v \in V(G)$ and $e \in E(G)$, $s_e$ does not intersect $p_v$ unless $v$ is an end-vertex of $e$. We can assume $G$ to be a 4-graph, as otherwise it does not have a rectilinear drawing. We will denote by $x(v)$ and $y(v)$ the $x$- and $y$-coordinates of $p_v$.

A \emph{vertex of $\Gamma$} is either a point that corresponds to a vertex of~$G$, in which case it is called a \emph{real-vertex}, or it is a point where two edges of~$G$ cross, in which case it is called a \emph{crossing-vertex}. 
If $\Gamma$ has no crossing-vertices, it is a \emph{planar rectilinear drawing}. An \emph{edge of $\Gamma$} is a portion of an edge of $G$ delimited by two vertices of $\Gamma$, which does not contain other vertices of $\Gamma$ in its interior. An edge of $\Gamma$ coincides with an edge $e$ of $G$ when $e$ does not cross any other edges of $G$ in $\Gamma$. 
Drawing $\Gamma$ divides the plane into connected regions, called \emph{faces}. The \emph{boundary} of each face $f$ is the circular sequence of vertices (either real- or crossing-vertices) and edges of $\Gamma$ that delimit $f$. The unique infinite region is the \emph{external face} of $\Gamma$; the other faces are the \emph{internal faces} of $\Gamma$. 
For a face $f$ of $\Gamma$, we denote by ${\cal B}(f)$ the boundary of $f$. If $f$ is the external face of $\Gamma$, we call ${\cal B}(f)$ the \emph{external cycle of} $\Gamma$.
%

\medskip
\noindent \textbf{Rotation systems.} A \emph{rotation system ${\cal R}(G)$ of $G$} specifies the clockwise order of edges in $E(v)$, for each vertex $v \in V(G)$. ${\cal R}(v)$ denotes the restriction of ${\cal R}(G)$ to $v$. A (rectilinear) drawing $\Gamma$ of $G$ determines a rotation system for $G$; in addition, $\Gamma$ determines the clockwise order of the edges incident to each crossing-vertex. For a given rotation system ${\cal R}(G)$, we say that $\Gamma$ \emph{preserves} ${\cal R}(G)$ if the rotation system determined by $\Gamma$ coincides with ${\cal R}(G)$. 

\medskip
\noindent \textbf{Our models.} We study \emph{\UERlong} (\emph{\UER drawings}), i.e., rectilinear drawings where all edges have unit length and all faces (including the external one) are rectangles.
Note that, in a \UER drawing, the external cycle only contains real-vertices.
Within the \UER model we also define a more restricted model, \emph{\UERIIlong} (\emph{\UERII drawings}), where each face is a square of unit side; see \Cref{fig:crd-examples}. The next simple~property~holds.
%
%

%





\begin{property}\label{pr:external-face}
Let $\Gamma$ be a \UER drawing of a graph $G$, let $C$ be the external cycle of $\Gamma$, and let $c_1, c_2, c_3, c_4$ be the vertices of $C$ at the corners of $\Gamma$. Then: $(i)$ $c_1, c_2, c_3, c_4$ have degree 2 in $G$; $(ii)$ each other vertex of $C$ has degree at most 3 in $G$; $(iii)$ the path in $C$ between two consecutive corners has the same length as the path in $C$ between the other two corners.  
\end{property}

We will assume that the input graph $G$ is biconnected and it is not a cycle, as
biconnectivity is necessary for the existence of \UER drawings, and if $G$ is a cycle the test is trivial: an \UERII drawing exists iff $G$ is a $4$-cycle; an \UER drawing exists iff $G$ has even length. 

\section{Unit Square Faces}\label{se:uerii}

In this section we focus on \UERII drawings. In addition to \Cref{pr:external-face}, we give two other properties of this model; see \Cref{fig:ufr-deg-4-vertices}. The first one 
comes from the fact that unit-square faces cannot contain an angle larger than $90^\circ$.

\begin{property}\label{pr:ufr-deg-4-vertices}
Let $\Gamma$ be a \UERII drawing and let $v$ be a real-vertex of $\Gamma$. If $v$ is not on the external cycle of $\Gamma$ then $\deg_\Gamma(v)=4$. Otherwise, $v$ is either a corner of $\Gamma$ or $\deg_\Gamma(v) = 3$. 
\end{property}

By \Cref{pr:ufr-deg-4-vertices}, the external cycle $C$ of any \UERII drawing has all vertices of degree 3, except for exactly four (the four corners), which have degree 2. 


\begin{property}\label{pr:ufr-paths}
A \UERII drawing contains disjoint vertical (horizontal) paths, connecting each non-corner vertex of the top (left) side with a non-corner vertex of the bottom (right) side, possibly traversing internal (degree-4) real-vertices. 
\end{property}

\begin{figure}[tb]
    \centering
    \begin{subfigure}{0.25\textwidth}
            \includegraphics[width=\textwidth,page=5]{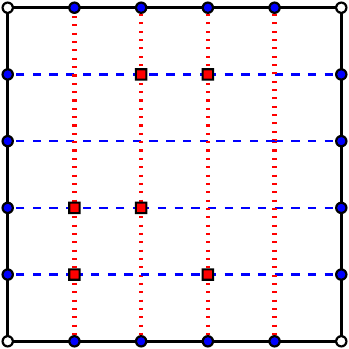}
            \subcaption{}
            \label{fig:ufr-deg-4-vertices}
    \end{subfigure}
    \hfil
    \begin{subfigure}{0.25\textwidth}
            \includegraphics[width=\textwidth,page=7]{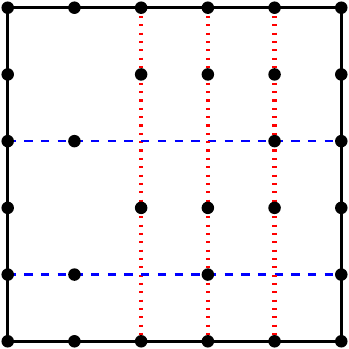}
            \subcaption{}
            \label{fig:uerr-proposition}
    \end{subfigure}
    \hfil
    \begin{subfigure}{0.25\textwidth}
            \includegraphics[width=\textwidth,page=9]{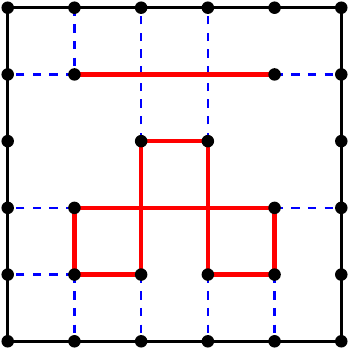}
            \subcaption{}
            \label{fig:uerr-inner-2-graph}
    \end{subfigure}
    \caption{(a) A \UERII drawing; internal vertices are red squares and external non-corner vertices are blue circles. (b) A \UER drawing without internal degree-3 vertices. 
    (c) A \UER drawing such that removing the external vertices yields a collection of paths and cycles (red and bold).}
    \label{fig:old-figs-2-4-6}
\end{figure}

We exploit \Cref{pr:external-face,pr:ufr-deg-4-vertices,pr:ufr-paths} to devise a linear-time recognition and construction algorithm for \UERII drawings. The algorithm can be adapted to the cases in which the external~cycle of the drawing and/or the rotation scheme of the graph are given. We~prove~the~following.

\begin{theorem}\label{th:uerii}
Let $G$ be an $n$-vertex $4$-graph. There exists an $O(n)$-time algorithm that tests whether $G$ admits a \UERII drawing. If such a drawing exists, the algorithm constructs one. Moreover, the algorithm can be adapted to preserve a given rotation system and/or to have a prescribed external cycle.
\end{theorem}

To prove \Cref{th:uerii} we describe an algorithm that attempts to construct a \UERII drawing $\Gamma$ of $G$ in two steps: First, it detects the external cycle $C$ and draws a rectangle $\Gamma_C$ representing $C$ (\Cref{subsec:usf-external-face}); then, it places the remaining vertices of $G$ in the interior of $\Gamma_C$ (\Cref{subsec:usf-internal}). If any of these two steps fails, the algorithm rejects the instance.

\subsection{Choosing and drawing the external face}\label{subsec:usf-external-face}

By \Cref{pr:external-face,pr:ufr-deg-4-vertices}, the external cycle $C$ must be composed only of degree-3 vertices, except for exactly four degree-2 vertices; also, the vertices of $G \setminus C$ must have degree 4. To find $C$, we first check whether $G$ has exactly four degree-2 vertices, which we call \emph{corners}, and then  remove all the degree-4 vertices from $G$; we denote by $G'$ the resulting graph. 

Since $C$ must be a spanning subgraph of $G'$, the graph $G'$ must be biconnected; if this is not the case, we can reject the instance. Note that $G'$ consists of the four corners, some other degree-2 vertices (whose third neighbor in $G \setminus G'$ has degree 4 in $G$), and some degree-3 vertices. For the degree-2 vertices of $G'$, both their incident edges must be part of $C$. For each degree-3 vertex $v$ of $G'$, we must decide which pair of edges incident to $v$ will belong to $C$ and which one will be a chord of $C$. If $G'$ has no degree-3 vertex, we can set $C=G'$. 

If $G'$ has some degree-3 vertices, we smooth the degree-2 vertices of $G'$ that are not corners; the resulting graph $G''$ is still biconnected and has all degree-3 vertices except the four corners.
We then decide the order in which the corners appear along $C$. Note that, if two corners are adjacent in $G''$, they must be consecutive along $C$. We guess each of the possible circular orders that is consistent with the possible adjacencies. There are at most three circular orders to consider, as we do not need to distinguish between a clockwise and an equivalent counterclockwise order. 

Let $c_1, c_2, c_3, c_4$ be the four corners as they appear in one of the guessed orders. We add to $G''$ the \emph{corner edges} $(c_1, c_2)$, $(c_2, c_3)$, $(c_3, c_4)$, and $(c_4, c_1)$, if not already present.
We claim that, for this order to be compatible with a \UERII drawing, $G''$ must have a special structure, which we will exploit to find $C$. Namely, let $C^*$ be any external cycle of a \UERII drawing of $G$ in which the four corners are $c_1, c_2, c_3, c_4$ in this circular order. Let $P_i$ be the path along $C^*$ between $c_i$ and $c_{i+1}$, for $i \in \{1,2,3\}$, and let $P_4$ the one between $c_4$ and $c_1$. By \Cref{pr:ufr-paths}, all chords of $C^*$ connect either a vertex of $P_1$ to a vertex of $P_3$, or a vertex of $P_2$ to a vertex of $P_4$ (see \Cref{fig:planar}). 
%
This implies that $G''$ must be planar and triconnected, so if this is not the case we reject the current guess. Otherwise, $G''$ admits a unique planar rotation system, up to a flip. In a planar embedding preserving this rotation system, the cycle $C$ to be selected as external cycle must be a concatenation of four paths, each sharing a face with a corner edge. This reduces the number of possible cycles to test to at most 16 (two candidate faces for each corner edge). However, each of the two faces incident to a corner edge shares an edge with a face incident to another corner edge. Thus, we must test only two disjoint cycles. \Cref{fig:two-ext-faces} shows an example in which both such cycles are valid; the two drawings have the same rotation~system.

\begin{figure}[tb]
    \centering
    \begin{subfigure}{0.3\textwidth}
            \includegraphics[width=\textwidth,page=6]{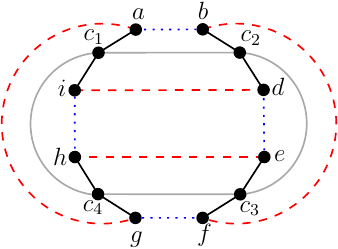}
            \subcaption{}
            \label{fig:planar}
    \end{subfigure}
    \hfil
    \begin{subfigure}{0.22\textwidth}
            \includegraphics[width=\textwidth,page=4]{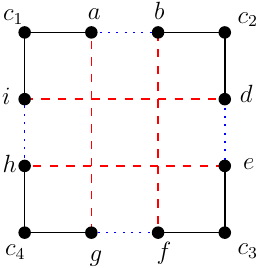}
            \subcaption{}
            \label{fig:ext-face-1}
        \end{subfigure}
        \hfil
        \begin{subfigure}{0.22\textwidth}
            \includegraphics[width=\textwidth,page=5]{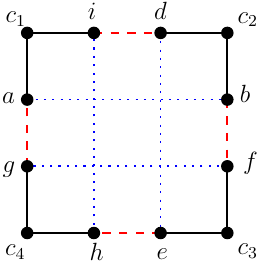}
            \subcaption{}
            \label{fig:ext-face-2}
        \end{subfigure}
    \caption{A graph admitting two possible \UERII drawings, with the same rotation system but different  external boundary. Corner edges are  grey.}
    \label{fig:two-ext-faces}
\end{figure}


To conclude, for each of the three possible guessed orders of corners we have obtained at most two candidate cycles of $G$. Let $C$ be one of such cycles. We check whether $C$ satisfies  Condition $(iii)$ of \Cref{pr:external-face} in linear time by traversing $C$; if not, we remove $C$ from the list of candidates. Otherwise, we draw $C$ as a rectangle with the four degree-2 vertices as corners and with all edges having unit length. Then, for each of the constantly many rectangles obtained in this step, we run the algorithm described in \Cref{subsec:usf-internal}.

\subsection{Drawing the internal vertices}\label{subsec:usf-internal}

Let $\Gamma_C$ be a rectangle representing a cycle $C$ obtained from the algorithm in \Cref{subsec:usf-external-face}. To place the vertices of $G \setminus C$ in the interior of $\Gamma_C$, we aim to construct the disjoint vertical/horizontal paths between pairs of opposite vertices of $C$ dictated by \Cref{pr:ufr-paths}. 
To this aim, we consider the vertices on the top side of $\Gamma_C$ in left-to-right order. When considering one of them, call it $u$, we assume as an invariant that all the vertical paths with $x$-coordinate smaller than $x(u)$ have already been drawn. 
Moreover, for each vertex $v$ already placed to the left of $u$, we assume that we have assigned a neighbor of $v$ to each direction (top/bottom/left/right)
including its right and bottom neighbors, even if they have not been placed yet. When we assign a vertex~$v$ to be the top/bottom/left/right neighbor of another vertex~$w$, we implicitly assume that~$w$ is assigned as the bottom/top/right/left neighbor of~$v$. As such, some of the vertices that have not been placed yet may have their left and/or top neighbor already assigned. We proceed as follows.

\medskip
\noindent \textbf{First Step.} 
We consider the vertices of $C$ that lie along the left side of $\Gamma_C$; let $v$ be one of them. Clearly, $v$ will not have any left neighbor and its top/bottom neighbors must be its immediate predecessor/successor along $C$ (walking on $C$ counterclockwise), so we can assign its unique remaining neighbor $w$ as its right neighbor. If $w$ belongs to $C$ and is not the unique vertex on the right side of $\Gamma_C$ with $y(w)=y(v)$, we reject $\Gamma_C$. Otherwise, the assignment satisfies the invariant.

\medskip
\noindent \textbf{General Step.}
We consider the leftmost vertex $u$ on the top side of $\Gamma_C$ that has not been considered yet. It has no top neighbor, whereas its left and right neighbors must be its immediate predecessor/successor along~$C$. So, we have to assign its other neighbor $w$ as its bottom neighbor. If the top neighbor of $w$ had already been assigned, or if $w$ belongs to $C$ and is not the unique vertex on the bottom side of $\Gamma_C$ with $x(w)=x(v)$, we reject~$\Gamma_C$. 

If $w \in C$, and thus lies on the bottom side of $\Gamma_C$, it has no bottom neighbor and its left/right neighbors are determined by $C$, so the invariant is satisfied. Having completed the vertical path starting at $u$, the algorithm restarts from the next vertex on the top~side. 

If $w \notin C$, we describe how to check and assign the neighbors of $w$ distinct from $u$. By the invariant, the left neighbor of $w$ has been already decided in a previous step. If this is not the case, we can reject $\Gamma_C$, otherwise we can place $w$ at $x$-coordinate $x(u)$ and $y$-coordinate equal to its left neighbor.
It remains to assign the other two neighbors of $w$, call them $w_1$ and $w_2$, as bottom and right neighbors of $w$. If one of them has already been placed, then it must be either the vertex of~$C$ on the bottom side of $\Gamma_C$ with $x$-coordinate $x(w)$, or the vertex of~$C$ on the right side of $\Gamma_C$ with $y$-coordinate $y(w)$, as otherwise we can reject~$\Gamma_C$. In any positive case, we can uniquely identify the bottom and right neighbors of $w$.
Assume vice versa that neither $w_1$ nor $w_2$ has already been placed. The key observation is that the vertex to be chosen as the bottom neighbor of $w$ must have its left neighbor already assigned, by the invariant. On the contrary, the vertex to be chosen as right neighbor should not have its left neighbor assigned, as we need to assign $w$ to this role. We check whether exactly one of $w_1$ and $w_2$ has its left neighbor assigned, and reject $\Gamma_C$ otherwise. If so, we can uniquely assign the bottom and right neighbors of $w$. Hence, the invariant holds in all cases. The algorithm continues with the one, $w_1$ or $w_2$, that has been selected as the bottom neighbor of $w$, which is processed in the same way as $w$.

\medskip

\noindent{\bf Proof of \Cref{th:uerii}.}
If all vertices on the top side of $\Gamma_C$ have been processed without rejecting $\Gamma_C$, we have a \UERII drawing $\Gamma$ of $G$, as each step of the algorithm fulfills the properties of \Cref{pr:external-face,pr:ufr-deg-4-vertices,pr:ufr-paths}. 
All rows and columns in~$\Gamma$ are covered by a path connecting two external degree-3 vertices, hence every
face of $\Gamma$ is a unit square. 
All operations described in \Cref{subsec:usf-external-face} can be performed in linear time, including the tests for planarity \cite{HopcroftT74,p-pte-13}, biconnectivity, and triconnectivity \cite{DBLP:journals/siamcomp/HopcroftT73}. 
Processing a vertex in the algorithm of \Cref{subsec:usf-internal} requires a constant number of operations on it and on its (at most four) neighbors. Thus, the algorithm runs in linear time. 
Finally, if the external cycle is prescribed, we check that it satisfies the required conditions, skipping the steps in \Cref{subsec:usf-external-face}. If the rotation system is given, we check that it is consistent with the unique rotation system of $G''$ in the algorithm of \Cref{subsec:usf-external-face}, and with the choices performed by the algorithm in \Cref{subsec:usf-internal}, which are unique in all~cases. 

\section{Rectangular Faces: Restricted Scenarios}\label{se:restricted-scenarios}
In this section we deal with \UER drawings and consider two different restricted scenarios for which we can provide efficient recognition and layout algorithms. In \Cref{se:uerr-no-internal-3} we focus on \UER drawings where the degree-3 vertices appear only on the external face. In \Cref{se:uerr-inner-2-graph} we take a further step towards the general case, by allowing internal degree-3 vertices, but requiring that the removal of the external cycle yields a collection of paths and cycles.

\subsection{No Internal Degree-3 Vertex}\label{se:uerr-no-internal-3}

Differently from \UERII drawings studied in \Cref{se:uerii}, this setting allows internal degree-2 vertices.  
Also, we may have more than four external degree-2 vertices, which makes the selection of the corners and the external cycle more challenging.
We start with a property that extends \Cref{pr:ufr-paths}, coming from the fact that each degree-4 (real- or crossing-) vertex  has four $90^\circ$ angles in its incident faces, whereas degree-2 vertices must have two $180^\circ$ angles; see \Cref{fig:uerr-proposition}.

\begin{property}\label{pr:uerr-paths}
A \UER drawing without internal degree-3 vertices contains disjoint vertical (horizontal) paths connecting each degree-3 vertex on the top (left) side with a degree-3 vertex on the bottom (right) side, possibly traversing internal degree-4 vertices. Furthermore, for each degree-2 vertex $v$ in a vertical (horizontal) path there is a degree-2 vertex in the left and right (top and bottom) sides of the rectangle representing the external cycle, excluding the corners, having the same $y$-coordinate ($x$-coordinate) as $v$.
\end{property}

In the following we exploit \Cref{pr:uerr-paths} to give several polynomial-time testing and construction algorithms for the problem, based on whether the corners, the external cycle, and/or the rotation system are part of the input.

\begin{theorem}\label{th:uerr-no-internal-3}
    Let $G$ be an $n$-vertex 4-graph. There exists a polynomial-time algorithm that tests whether $G$ admits a \UER drawing with no internal degree-3 vertex. If such a drawing exists, the algorithm constructs one. Moreover, the algorithm can be adapted to preserve a given rotation system and/or to have a prescribed external cycle.
    The time complexity of the algorithm is:
    \begin{enumerate*}
        \item $O(n)$ if the four corners are given;
        \item $O(n)$ if the rotation system is prescribed;
        \item $O(n)$ if $G$ is a 3-graph and the external cycle is prescribed;
        \item $O(n^3)$ if the external cycle is prescribed; and
        \item $O(n^5)$ in the general case.
    \end{enumerate*}
\end{theorem}


As in \Cref{se:uerii}, the algorithm supporting \Cref{th:uerr-no-internal-3} is composed of two steps, the first one to select and draw the external cycle, and the the second one to draw the internal vertices. In \Cref{subsec:uerr-no-internal-3-external-face} we describe several algorithms for the first step, depending on the specific setting. The algorithm to draw the internal vertices in the interior of a drawing of the external cycle is the same in all settings, and is described in \Cref{subsec:uerr-no-internal-3-drawing}.

\subsubsection{Choosing and drawing the external face.}\label{subsec:uerr-no-internal-3-external-face}

We start with the case in which the corners of the rectangle representing the external cycle are prescribed as part of the input (Item $1$ of \Cref{th:uerr-no-internal-3}).

\begin{restatable}{lemma}{leExternalCornersGiven}\label{le:external-corners-given}
Let $G$ be an $n$-vertex 4-graph and let the corners of $G$ be given. There exists an $O(n)$-time algorithm that constructs a constant number of rectangles such that if $G$ admits a \UER drawing with no internal degree-3 vertex and with the prescribed corners, then the drawing of the external cycle coincides with one of these rectangles.
\end{restatable}

\begin{proof}
We employ the same algorithm as the one in \Cref{subsec:usf-external-face}, using the four vertices prescribed in the input as the four corners. Namely, we remove the degree-4 vertices and smooth the degree-2 vertices different from the corners. This yields a graph $G''$ with all degree-3 vertices, except for the corners. We then guess all the possible orders of the four corners, add the edges between them according to this order, and test whether the resulting graph is triconnected and planar. If so, from the unique planar embedding, we read the at most two possible cycles that can be used as external cycle in  the desired \UER drawing. The only technical difference is that, in this case, the edges of $G''$ that are selected as chords may have been obtained by smoothing degree-2 vertices, while in \Cref{subsec:usf-external-face} these edges had to be present as edges also in the original graph $G$, since in that case the smoothed degree-2 vertices could not be placed in the interior of the drawing. However, in the algorithm in \Cref{subsec:usf-external-face} we did not exploit this property in this step, and deferred the discovery of this potential problem to \Cref{subsec:usf-internal}, where the internal vertices are placed. As such, the algorithm can be used without modification also in this case.
\end{proof}

We then consider the case where the external cycle $C$ is prescribed, but not the corners. 
In this case, by Condition (iii) of \Cref{pr:external-face}, deciding two adjacent corners allows us to infer the other two. Thus, one can guess $O(n^2)$ pairs of degree-2 vertices of $C$ and get the corresponding $O(n^2)$ candidate rectangles (Item $4$ of \Cref{th:uerr-no-internal-3}).
In the following we show that, if we further have that $G$ does not contain degree-4 vertices, then the rectangle corresponding to the given cycle (if any) is unique (Item $3$ of \Cref{th:uerr-no-internal-3}). 

\begin{lemma}\label{le:external-deg3-cycle-given}
    Let $G$ be an $n$-vertex 3-graph and let the external cycle $C$ of $G$ be given. There exists an $O(n)$-time algorithm that constructs a rectangle $\Gamma_C$ such that, if $G$ admits a \UER drawing with no internal degree-3 vertex and with $C$ as external cycle, then the drawing of the external face coincides with $\Gamma_C$.
\end{lemma}
\begin{proof}
    Since $G$ has no degree-4 vertex and all its degree-3 vertices are on the external cycle~$C$, we have that $G$ consists of $C$ and of a set $\mathcal{P}$ of disjoint paths connecting pairs of vertices of $C$. Our goal is to represent these paths as vertical or horizontal paths satisfying \Cref{pr:uerr-paths}.
    We prove that there is only one possible drawing $\Gamma_C$ of $C$ in any \UER drawing of $G$ whose external boundary coincides with $C$, up to rotation or a flip of the entire drawing, unless $G=C$.

    Let $v$ be a degree-3 vertex. By \Cref{pr:external-face}, $v$ cannot be a corner of $\Gamma_C$. Let $P \in {\cal P}$ be the path incident to $v$, let $u$ be the other endpoint of $P$, and let $C_L$ and $C_R$ be two paths along $C$ between $u$ and $v$; refer to \Cref{fig:find-external-face-a}. Observe that each of $C_L$ and $C_R$ must contain two of the corners of $\Gamma_C$, since $u$ and $v$ must lie on opposite sides of $\Gamma_C$, by \Cref{pr:uerr-paths}.
    We assume, w.l.o.g., that there exist no two vertices of $C_R$ that are connected by a path in $\mathcal{P}$, as otherwise we could select those two vertices as $u$ and $v$.
    Let $k$ be the number of degree-3 vertices in $C_R$. By the previous assumption, each of the $k$ vertices is the endpoint of a path in $\mathcal{P}$ that crosses $P$ in any \UER drawing of $G$.


    \begin{figure}[tb]
        \centering
        \begin{subfigure}[b]{0.35\textwidth}
            \includegraphics[width=\textwidth,page=1]{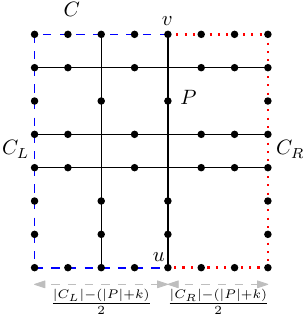}
            \subcaption{}
            \label{fig:find-external-face-a}
        \end{subfigure}
        \hfil
        \begin{subfigure}{0.35\textwidth}
            \includegraphics[width=\textwidth,page=2]{./figures/find-external-face.pdf}
            \subcaption{}
            \label{fig:find-external-face-b}
        \end{subfigure}
        \caption{Illustrations for: (a) \Cref{le:external-deg3-cycle-given} ($k=3$, $|P|=5$);  
        (b) \Cref{le:external-rotation-given} ($k=4$, $|P|=6$, $|P|_4=2$, $|P|_{\neq 4}=4$). $C_L$ is blue dashed, $C_R$ is red dotted, $P$ is bold.}
        \label{fig:find-external-face}
    \end{figure}

    Thus, by \Cref{pr:uerr-paths}, the number of vertices on one side of $\Gamma_C$, say the right side, including the corners, must be equal to $|P|+k$, where $|P|$ denotes the number of vertices of~$P$. Thus, the corners of $\Gamma_C$ in $C_R$ must be at distance $\frac{|C_R|-(|P|+k)}{2}$ from $u$ and $v$, respectively, $|C_R|$ being the number of vertices of $C_R$. Also, the corners of $\Gamma_C$ in $C_L$ must be at distance $\frac{|C_L|-(|P|+k)}{2}$ from $u$ and $v$. 
    If one of the vertices chosen as corners has degree~3, we reject the instance by \Cref{pr:external-face}. Else, we realize $\Gamma_C$ as a unit edge-length rectangle.
\end{proof}

We now consider the case in which the rotation system of $G$ is prescribed, while the external cycle and the corners are not (Item $2$ of \Cref{th:uerr-no-internal-3}). The strategy is to identify at most three cycles of $G$ as potential external cycles, and then apply a strategy similar to \Cref{le:external-deg3-cycle-given} to find the corners. In particular, for degree-4 vertices we can use the rotation system to understand the direction of the paths passing through them. We have the following.

\begin{restatable}{lemma}{leExternalRotationGiven}\label{le:external-rotation-given}
 Let $G$ be an $n$-vertex 4-graph and let the rotation system $\mathcal{R}(G)$ of $G$ be given. There exists an $O(n)$-time algorithm that constructs a constant number of rectangles such that if $G$ admits a \UER drawing with no internal degree-3 vertex and with rotation system~$\mathcal{R}(G)$, then the drawing of the external face coincides with one of these rectangles.  
\end{restatable}
\begin{proof}
First observe that, if $G$ has no degree-3 vertex, then all the external vertices in any \UER drawing of $G$ will have degree 2 in $G$, which is not possible unless $G$ is a cycle (or the graph is disconnected). Assume then that $G$ has at least one degree-3 vertex, and let $v$ be such a vertex.
By \Cref{pr:external-face} and by the requirements of the lemma, $v$ must be an external vertex distinct from the corners in any \UER drawing of $G$. In particular, one of the three faces incident to $v$ must coincide with the external face. We guess each of the three faces and test them separately; this corresponds to guessing the two edges of $C$ incident to $v$.

Let $(v,w)$ and $(v,z)$ be the two edges incident to $v$ assigned to $C$, such that $(v,w)$ follows $(v,z)$ clockwise in $\mathcal{R}(v)$. If $w$ has degree 4, we reject the current guess, by \Cref{pr:external-face}. If $w$ has degree 2, we assign also its other incident edge to $C$ and move on to analyzing this edge. If $w$ has degree 3, we can use the rotation system of $w$ to select the other edge incident to $w$ that is in $C$; namely, we add to $C$ the edge following $(w,v)$ clockwise in $\mathcal{R}(w)$, and move on to analyzing this edge. When we reach a vertex whose edges have already been assigned to $C$, we check whether it coincides with $v$; if not, we reject the current guess as the constructed cycle $C$ must be a simple cycle. Moreover, we check whether there are still degree-3 vertices in $G \setminus C$, in which case we reject the current guess. If we have not rejected the guess, we have obtained a simple cycle $C$ containing all the degree-3 vertices.
   
In order to select the four corners among the vertices of $C$, we use the same strategy as in \Cref{le:external-deg3-cycle-given}, namely we search for a set of paths $\mathcal{P}$ that connect pairs of vertices of $C$, and we use one of them to find the only selection of corners, if any, that complies with \Cref{pr:uerr-paths}. 

Note that, in this case, the paths in $\mathcal{P}$ are not necessarily disjoint, due to the presence of degree-4 vertices. However, in this case, we can assign to the same path each pairs of edges that are not consecutive in the rotation system around a degree-4 vertex, since one of these paths will be drawn vertical and the other horizontal. This implies that the set $\mathcal{P}$ can still be constructed in linear time starting from the degree-3 vertices on $C$. 

To compute the corners, let $P \in {\cal P}$ be the path between two vertices $u$ and $v$ such that one of the paths $C_R\subset C$ connecting $u$ and $v$ does not contain two vertices connected by another path in $\mathcal{P}$; refer to \Cref{fig:find-external-face-b}. Let $k$ be the number of degree-3 vertices in $C_R$, let $|P|_4$ be the number of degree-4 vertices of $P$, and let $|P|_{\neq 4} = |P| - |P|_4$. By \Cref{pr:uerr-paths}, each degree-2 vertex of $P$ must be aligned with a degree-2 vertex of the right/left side of $\Gamma_C$, and each degree-4 vertex of $P$ must be aligned with a degree-3 vertex of the right/left side of $\Gamma_C$.
Hence, if $k < |P|_4$ we reject the current guess; otherwise, the right side of $\Gamma_C$ must contain $|P|_{\neq 4} + k$ vertices, so two of the corners must be at distance $\frac{|C_R|-(|P|_{\neq 4}+k)}{2}$ from $u$ and $v$, respectively. The other two corners are computed by relying on Condition (iii) of \Cref{pr:external-face}.
\end{proof}

Finally, consider the general case in which no additional information is prescribed in the input (Item $5$ of \Cref{th:uerr-no-internal-3}). In this case, we guess the $O(n^4)$ combinations of the four corners, and apply \Cref{le:external-corners-given} to obtain $O(n^4)$ candidate rectangles in $O(n^5)$ time.

\subsubsection{Drawing the internal vertices.}\label{subsec:uerr-no-internal-3-drawing}

We now describe how to test whether one of the rectangles $\Gamma_C$ constructed in \Cref{subsec:uerr-no-internal-3-external-face} to represent the external cycle $C$ can be completed to a \UER drawing of $G$, by placing the vertices of $G \setminus C$ in its interior.
Recall that the vertices of $G \setminus C$ have degree either~2 or~4. Since the two edges incident to a degree-2 vertex must be drawn both vertical or both horizontal in a \UER drawing, we smooth such vertices and get a graph with all the internal vertices of degree 4. Thus, we can apply the same algorithm in \Cref{subsec:usf-internal} to compute the top/bottom/left/right neighbor of each degree-4 vertex. Recall that this assignment of neighbors is unique, if it exists. We then restore the degree-2 vertices and check if, consistently with  \Cref{pr:uerr-paths}, the resulting  $x$- and $y$-coordinates of the degree-2 vertices coincide with those of the degree-2 vertices on the sides of $\Gamma_C$. 

\bigskip

\noindent{\bf Proof of \Cref{th:uerr-no-internal-3}}. The correctness follows from \Cref{pr:uerr-paths} and \Cref{le:external-corners-given,le:external-deg3-cycle-given,le:external-rotation-given}, for the corresponding three cases, and from the exhaustive guesses in the other two. The time complexity comes from applying the linear-time algorithm described in \Cref{subsec:uerr-no-internal-3-drawing} to each of the candidate rectangles constructed in the different cases, and since their number is bounded by a constant.

\subsection{Internal Paths and Cycles}\label{se:uerr-inner-2-graph}

We now consider the restricted scenario in which the input 4-graph $G$ contains a cycle $C$ such that $G \setminus C$ has vertex-degree at most two, i.e., it is a collection of paths and cycles (see, e.g., \Cref{fig:uerr-inner-2-graph}). We call $G$ an \emph{inner-2-graph with respect to $C$}. We only study the case when $C$ is given as part of the input. 


\begin{restatable}{theorem}{thInnerTwoGraph}\label{th:inner-2-graph}
    Let $G$ be an inner-2-graph with respect to a given cycle $C$. There exists a polynomial-time algorithm that tests whether $G$ admits a \UER drawing whose external cycle coincides with $C$. If such a drawing exists, the algorithm constructs one. If the four corners are prescribed, the algorithm takes $O(n^2)$ time, otherwise it takes $O(n^4)$ time. Furthermore, the algorithm can be adapted to preserve a given rotation system.
\end{restatable}

%

\noindent \textbf{Algorithm.}
We first compute all rectangles that are candidate to represent the external boundary of the drawing.  
%
Let $\Gamma_C$ be one of the candidate rectangles computed in the previous step. We show how to place the vertices of $G \setminus C$ in the interior of $\Gamma_C$.
Assume, w.l.o.g., that the bottom-left corner of $\Gamma_C$ has coordinates $(0,0)$. Let $W$ and $H$ be the maximum $x$ and $y$ coordinates of $\Gamma_C$, respectively. We traverse the points $(i,j)$ ($1 \leq i \leq W-1$ and $1 \leq j \leq H-1$) of the grid that lie internally to $\Gamma_C$ from left to right and secondarily from top to bottom starting from the top-leftmost one, which has coordinates $(1,H-1)$.
Since we process the grid points from left to right and from top to bottom, the $x$-coordinate $i$ increases and the $y$-coordinate $j$ decreases. In the next description, we call \emph{placed vertices} those vertices whose coordinates have already been assigned by the algorithm (i.e., the vertices of $C$ and those already placed at some internal grid points). When processing the point $(i,j)$ we call \emph{fixed vertices} the placed vertices with coordinates $(i',j')$ such that either $i'<i$ or $i'=i$ and $j'>j$ (see \Cref{fig:inner-2-graphs.a}).

    \begin{figure}[ptb]
        \centering
        \begin{subfigure}[b]{0.48\textwidth}
            \includegraphics[width=\textwidth,page=1]{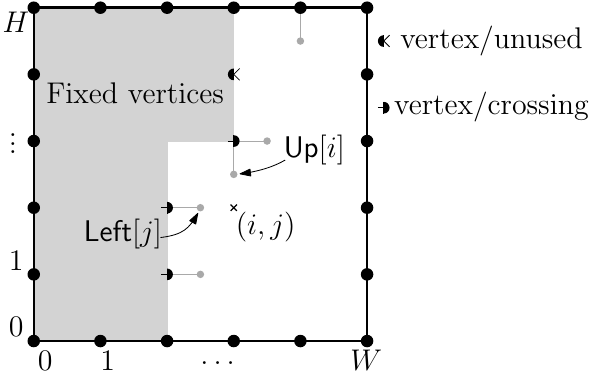}
            \subcaption{}
            \label{fig:inner-2-graphs.a}
        \end{subfigure}
        \hfil
        \begin{subfigure}{0.48\textwidth}
            \includegraphics[width=\textwidth,page=2]{./figures/inner-2-graphs.pdf}
            \subcaption{}
            \label{fig:inner-2-graphs.b}
        \end{subfigure}
        \hfil
        \begin{subfigure}{0.48\textwidth}
            \includegraphics[width=\textwidth,page=3]{./figures/inner-2-graphs.pdf}
            \subcaption{}
            \label{fig:inner-2-graphs.c}
        \end{subfigure}
        \hfil
        \begin{subfigure}{0.48\textwidth}
            \includegraphics[width=\textwidth,page=4]{./figures/inner-2-graphs.pdf}
            \subcaption{}
            \label{fig:inner-2-graphs.d}
        \end{subfigure}
        \hfil
        \begin{subfigure}{0.48\textwidth}
            \includegraphics[width=\textwidth,page=5]{./figures/inner-2-graphs.pdf}
            \subcaption{}
            \label{fig:inner-2-graphs.e}
        \end{subfigure}
        \hfil
        \begin{subfigure}{0.48\textwidth}
            \includegraphics[width=\textwidth,page=6]{./figures/inner-2-graphs.pdf}
            \subcaption{}
            \label{fig:inner-2-graphs.f}
        \end{subfigure}
        \caption{Illustration for \Cref{th:inner-2-graph}. (a) Notation of the algorithm. (b)--(f) Processing the point $(i,j)$:   (b) Case 1; (c) Case 2; (d) Case 3; (e) Case 4; (f) Case 5.}
        \label{fig:hortalemma}
    \end{figure}
    
    During the traversal we maintain two arrays of vertices \upA and \leftA of size $W$ and $H$, respectively; see \Cref{fig:inner-2-graphs.a}. Intuitively, they are used to extend the concept of top/left/bottom/right neighbors to all grid points, even if they do not contain a vertex. More specifically, when we process the point $(i,j)$, let $q$ be the rightmost fixed vertex such that $y(q)=j$; analogously, let $r$ be the bottommost fixed vertex such that $x(r)=i$.  If $\leftA[j]$ stores a vertex $q'$, then $q'$ is adjacent to $q$ and $y(q')$ will be equal to $j$ whereas $x(q')$ has not been decided yet; in other words, it is already decided that $q'$ will be horizontally aligned with $q$ (and thus it is candidate to occupy point $(i,j)$). If $\leftA[j]$ is \nullA, then no neighbor of $q$ will be assigned $y$-coordinate $j$; this implies that $(i,j)$ cannot be occupied by a crossing.  Analogously, if $\upA[i]$ stores a vertex $r'$, then it is already decided that $r'$ will be vertically aligned with $r$; if $\upA[i]$ is \nullA, then no neighbor of $r$ will be assigned $x$-coordinate $i$.  

    We initialize the array \leftA as follows. For $1\leq j \leq H-1$, let $v_j$ be the vertex of the left side of $\Gamma_C$ such that $y(v_j)=j$. If $\deg_G(v_j)=2$, we set $\leftA[j]=\nullA$. If $\deg_G(v_j)=3$, we set $\leftA[j]=u_j$, where $u_j$ is the neighbor of $v_j$ that is not on the left-side of $\Gamma_C$. The vertex $u_j$ is either on the right side of $\Gamma_C$ or it is an internal vertex of $G$. If $u_j\in C$ and $y(u_j) \neq y(v_j)$, we reject $\Gamma_C$.
    The $\upA$ array is initialized similarly, considering the vertices on the top side of $\Gamma_C$, and their neighbors that do not lie on the top side of~$\Gamma_C$ (if any).
    
    We now explain how to process the generic point $(i,j)$.  
    We distinguish 5 cases depending on the values in the \leftA and \upA arrays (see also \Cref{fig:inner-2-graphs.b,fig:inner-2-graphs.c,fig:inner-2-graphs.d,fig:inner-2-graphs.e,fig:inner-2-graphs.f}):
    \begin{enumerate*}
        \item $\leftA[j]=\upA[i]$ and they are both not \nullA;
        \item $\leftA[j]\neq \upA[i]$ and they are both not \nullA;
        \item $\leftA[j]\neq\nullA$ and $\upA[i]=\nullA$;
        \item $\leftA[j]=\nullA$ and $\upA[i]\neq\nullA$; and
        \item $\leftA[j]=\upA[i]$ and they are both \nullA.
    \end{enumerate*}
    In Case 1, let $v$ be the vertex stored in $\leftA[j]=\upA[i]$. We map $v$ to the point $(i,j)$ and update the two arrays as explained below.
    In Case 2, we put a crossing in the point $(i,j)$ and leave \leftA and \upA unchanged.
    In Case 3, let $v=\leftA[j]$; we map $v$ to $(i,j)$ and update the two arrays as described below.
    Similarly, in Case 4, let $v=\upA[i]$; we map $v$ to $(i,j)$ and update the two arrays as described below.
    Finally, in Case 5 the point $(i,j)$ will be left unused (it will be internal to a face), and we leave \leftA and \upA unchanged.
    
    The updates of the values $\leftA[j]$ and $\upA[i]$ in Cases 1, 3, and 4, are done as follows. Let $U$ be the set of neighbors of $v$ that are not fixed. 
    If $|U|=0$, we reject the instance, as either $\deg_G(v)=2$ and its incident edges form a $270^\circ$ angle (Case 1), or $\deg_G(v)=1$ (Cases 3 and 4). Also, if $|U| \geq 3$, we reject the instance as in this case there is no possibility of placing all vertices in $U$. 
    Thus it must be $|U| \in \{1,2\}$; let $u$ and $u'$ be the two vertices of $U$, possibly with $u'=\nullA$ if $|U|=1$. In this case, for each vertex $u$ and $u'$, either its coordinates are unassigned or it belongs to $C$.  
    Suppose first that at least one of them, say $u$, belongs to $C$; if $x(v) \neq x(u)$ and $y(v) \neq y(u)$, we reject the instance because $v$ and $u$ cannot be horizontally or vertically aligned. If $x(v) = x(u)$ we set $\upA[i]=u$ and $\leftA[j]=u'$; if $y(v) = y(u)$ we set $\upA[i]=u'$ and $\leftA[j]=u$. If $u'$ exists and it is already placed, then we further check that its coordinates are consistent with its assignment to $\upA[i]$ or $\leftA[j]$.   

    Suppose now that neither $u$ nor $u'$ is in $C$. 
    If $\deg_G(u)\geq 3$ we set $w{=}u$; else we follow the edge incident to $u$ distinct from $(u,v)$ and continue traversing all degree-2 vertices until we reach a degree-3 vertex $w$. If $w$ is fixed, we reject the instance, as either the traversed path can only be drawn with an angle of $270^\circ$, or $u$ has to be on the left of $v$, contradicting the fact that it is not fixed. Hence $w$ is non-fixed and either it belongs to $C$ or its coordinates are unassigned. Assume first that $w \in C$. If $x(w){\neq} x(v)$ and $y(w){\neq} y(v)$ we reject the instance, as the path visited by the traversal must be either horizontal or vertical. Else, if $x(w){=}x(v)$ we set $\upA[i]{=}u$, and $\leftA[j]{=}u'$. Similarly, if $y(w){=}y(v)$ we set $\leftA[j]{=}u$ and $\upA[i]{=}u'$. 
    If the coordinates of $w$ have not been assigned, then $w \in G\setminus C$. Since $G \setminus C$ is a 2-graph and $\deg_G(w)>2$, at least one neighbor of $w$, call it $z$, belongs to $C$. 
    If $z$ lies on the left or right side of $\Gamma_C$, and $0<y(z)<y(v)$ or $x(z)=x(v)$ and $y(z)=0$, we set
    $\upA[i]=u$ and $\leftA[j]=u'$ (see \Cref{fig:hortalemma-down}).
    If $z$ lies on the top or bottom side of $\Gamma_C$, and $x(v)<x(z)<W$ or $x(z)=W$ and $y(z)=y(v)$, we set $\leftA[j]=u$ and $\upA[i]=u'$ (see \Cref{fig:hortalemma-right}).
    If $z$ does not lie in any of the described positions, we reject the instance, as $z$ cannot share a coordinate with $w$, and thus the path from $v$ to $w$ cannot be drawn horizontally or vertically.
    
    Once all the grid points have been considered, if there is a vertex of $G$ with unassigned coordinates, we reject the instance, otherwise we have a drawing with external boundary~$\Gamma_C$.

    \begin{figure}
        \centering
        \begin{subfigure}[b]{0.25\textwidth}
            \includegraphics[width=\textwidth,page=3]{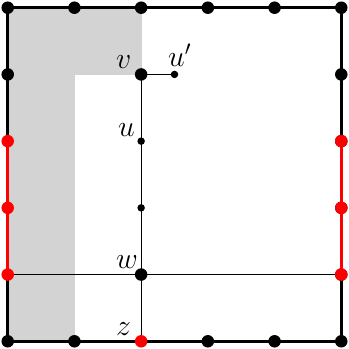}
            \subcaption{}
            \label{fig:hortalemma-down}
        \end{subfigure}
        \hfil
        \begin{subfigure}{0.25\textwidth}
            \includegraphics[width=\textwidth,page=5]{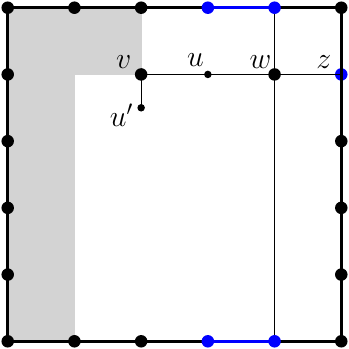}
            \subcaption{}
            \label{fig:hortalemma-right}
        \end{subfigure}
        \caption{(a) If $z$ coincides with one of the red points of $\Gamma_C$, then $x(u)=x(v)$. (b) If $z$ coincides with one of the blue points of $\Gamma_C$, then $y(u)=y(v)$.}
        \label{fig:hortalemma}
    \end{figure}

\noindent \textbf{Correctness.} 
    We prove the correctness of the algorithm by proving that the  following invariants hold when we process the grid point $(i,j)$:
    \begin{itemize}
        \item \textbf{I1}: The drawing induced by the placed vertices is a \UER drawing with the only exception of one face which may not be rectangular; such a face (if it exists) is the one containing all the grid points still unprocessed.
        \item  \textbf{I2}: If $\leftA[j] \neq \nullA$, then the vertex $v=\leftA[j]$ has to be drawn necessarily with $y$-coordinate equal to $j$; if $\leftA[j] = \nullA$, none of the neighbors of the rightmost fixed vertex $q$ with $y(q)=j$ can be drawn with $y$-coordinate $j$.
        \item \textbf{I3}: If $\upA[i] \neq \nullA$, then the vertex $v=\upA[i]$ has to be drawn necessarily with $x$-coordinate equal to $i$; if $\upA[i] = \nullA$, none of the neighbors of the bottommost fixed vertex $r$ with $x(r)=i$ can be drawn with $x$-coordinate $i$.
    \end{itemize}

    The invariants hold before we process the first grid point internal to $\Gamma_C$. Namely, the drawing induced by the placed vertices is $\Gamma_C$ and therefore it is a \UER drawing. Invariant \textbf{I2} holds because of the way we initialize the array \leftA. Namely, if $\deg_G(v_j)=2$, we set $\leftA[j]=\nullA$ and Invariant \textbf{I2} holds for $\leftA[j]$ because all the neighbors of $v_j$ belong to $\Gamma_C$; if $\deg_G(v_j)=3$, we set $\leftA[j]=u_j$ and Invariant \textbf{I2} holds for $\leftA[j]$ because the only neighbor $u_j$ of $v_j$ that does not belong to $\Gamma_C$ must be horizontally aligned with $v_j$. A symmetric argument applies to Invariant \textbf{I3}. 

    Assume now that the invariants hold before processing the grid point $(i,j)$.
    Because of invariants \textbf{I2} and \textbf{I3} the decision taken by the algorithm about point $(i,j)$ (Cases 1--5) is correct. Notice that each of the choices made in Cases 1--5 maintains Invariant \textbf{I1}.  We now prove that the updates of the value $\leftA[j]$ and $\upA[i]$ maintain the Invariants \textbf{I2} and \textbf{I3}. In Cases $2$ and $5$ $\leftA[j]$ and $\upA[i]$ are not changed and the Invariants \textbf{I2} and \textbf{I3} still hold. Consider Cases $1$, $3$, and $4$. In all the three cases a vertex $v$ is placed at point $(i,j)$ and $\leftA[j]$ and $\upA[i]$ are updated looking at the non-fixed neighbors $U$ of $v$. As already observed, if $|U|\in \{0,3,4\}$, the algorithm correctly rejects the instance. Otherwise, $v$ has at most two non-fixed neighbors $u$ and $u'$ (with $u'$ possibly equal to $\nullA$). If one of the two, say $u$, belongs to $C$, then $u$ and $v$ must be aligned either horizontally, if they share the $x$-coordinate, or vertically, if they share the $y$-coordinate; in the first case we correctly set $\leftA[j]=u$, while in the second case we correctly set $\upA[i]=u$ (if $u$ and $v$ do not share any coordinate we correctly reject the instance). If $u'$ is not \nullA, we correctly set $\upA[i]=u'$ if $\leftA[j]=u$, and viceversa. Notice that, if $u'$ is already placed we have to check that its coordinates allow the desired alignment; if not we reject the instance. 
If neither $u$ nor $u'$ belong to $C$, we search for a non-fixed vertex $w$ whose degree is at least $3$ such that $v$ and $w$ are connected by a chain of degree-2 non-fixed vertices, the first one being $u$ ($u$ and $w$ may coincide, in which case the chain is empty).  If $w$ belongs to $C$, then $v$, $w$ and all the vertices of the  chain connecting $v$ and $w$ must be aligned either horizontally, if $v$ and $w$ share the $y$-coordinate, or vertically, if they share the $x$-coordinate; in the first case we correctly set $\leftA[j]=u$, while in the second case we correctly set $\upA[i]=u$ (if $w$ and $v$ do not share any coordinate we correctly reject the instance). If $w$ does not belong to $C$, then it will be adjacent to a vertex $z$ that belongs to $C$.  In each of the three cases described above (see also \Cref{fig:hortalemma}) the position of $z$ in $\Gamma_C$ implies the vertical or horizontal alignment of $v$ and $u$, or the rejection of the instance. In each case we correctly set the values of $\leftA[j]$ and $\upA[i]$.
    
\noindent \textbf{Time Complexity.}  We conclude the proof by discussing the time complexity of our algorithm. 
We first compute all rectangles that are candidate to represent the external boundary of the drawing. The time complexity of this step will be discussed later.  

As each point of the grid that lies internally to $\Gamma_C$ is processed at each step, we have $O(n^2)$ steps. In fact, the grid can have size quadratic with respect to the number of vertices of $G$. Each step in which a vertex is placed consists of constant-time operations on the arrays and on the vertex, plus a traversal of the graph that visits a path of degree-2 vertices that connects two vertices with higher degree. We claim that each vertex $u$ is visited by at most four traversals. To prove the claim, it is sufficient to notice that each path of degree-2 vertices is traversed at most once. Thus, if $\deg_G(u)=2$, $u$ belongs to the path visited by a unique traversal, and if $\deg_G(u)\geq3$, then there are at most two traversal having $u$ as their starting point and two traversal having $u$ as ending point. Thus, given the candidate rectangle $\Gamma_C$, the algorithm runs in $O(n^2)$ time. If $\Gamma_C$ is not given we first compute all rectangles that are candidate to represent the external boundary of the drawing.  More precisely, if four vertices of $C$ are prescribed to be the corners of the external face, we have at most one rectangle respecting \Cref{pr:external-face}. Otherwise, as for Item 4 of \Cref{th:uerr-no-internal-3}, we guess $O(n^2)$ pairs of degree-2 vertices to be consecutive corners and infer the other two based on Condition (iii) of \Cref{pr:external-face}, thus obtaining $O(n^2)$ candidate rectangles.
Thus if the fours corners are given the algorithm runs in $O(n^2)$ time. Otherwise the algorithm runs in $O(n^4)$ time.

\section{Rectangular Faces: General Case}\label{se:uerr}

For the general case we describe an FPT algorithm in the number of degree-3 vertices (\Cref{th:uerr-fpt}). To this aim, we first provide polynomial-time algorithms when the angles around each degree-3 vertex are given in input (\Cref{le:uerr-fixed-angles}). More formally, 
a \emph{large-angle assignment} $\mathcal{A}(G)$ of a 4-graph $G$ determines for each degree-3 vertex $v$ of $G$ a pair $\mathcal{A}(v)=\langle u,w\rangle$, such that $u,w \in N(v)$. We say that a \UER drawing $\Gamma$ of $G$ \emph{realizes} a large-angle assignment $\mathcal{A}(G)$ if, for each degree-3 vertex $v$ of $G$, with $\mathcal{A}(v)=\langle u,w\rangle$, we have in $\Gamma$ that either $x(u)=x(w)$ or $y(u)=y(w)$, i.e. the angle delimited by the edges $(u,v)$ and $(v,w)$ is a $180^\circ$-degree angle. Note that $\mathcal{A}(G)$ does not imply a rotation system~for~$G$.

\begin{lemma}\label{le:uerr-fixed-angles}
    Let $G$ be a 4-graph with a large-angle assignment $\mathcal{A}(G)$.
    There exists a polynomial-time algorithm that tests whether $G$ admits an angle-preserving \UER drawing. If such a drawing exists, the algorithm constructs one. Moreover, the algorithm can be adapted to preserve a given rotation system and/or to have a prescribed external cycle.
    The time complexity of the algorithm is:
    \begin{enumerate*}
        \item $O(n^2)$ if the four corners are given;
        \item $O(n^4)$ if the external cycle is prescribed;
        \item $O(n^{4.5})$ in the general case.
    \end{enumerate*}
\end{lemma}

We present the algorithm supporting \Cref{le:uerr-fixed-angles} in \Cref{subsec:uer-external-face,subsec:uer-drawing}, describing the choice and drawing of the external cycle and the placement of the internal vertices, respectively.
%
%
By exploiting \Cref{le:uerr-fixed-angles} we can obtain an FPT algorithm for the general case, parametrized by the number of degree-3 vertices, by guessing for each degree-3 vertex a pair of incident edges to form the large angle, and apply the algorithm of \Cref{le:uerr-fixed-angles}. We have the following.

\begin{theorem}\label{th:uerr-fpt}
    Let $G$ be a 4-graph and let $k$ be the number of degree-3 vertices $G$. There exists an FPT algorithm, with parameter $k$, that tests whether $G$ admits a \UER drawing. If such a drawing exists, the algorithm constructs one. Moreover, the algorithm can be adapted to preserve a given rotation system and/or to have a prescribed external cycle.
    The time complexity of the algorithm is:
    \begin{enumerate*}
        \item $O(3^k n^2)$ if the four corners are given;
        \item $O(3^k n^4)$ if the external cycle is prescribed;
        \item $O(3^k n^{4.5})$ in the general case.
    \end{enumerate*}
\end{theorem}

\subsection{Choosing and drawing the external face}\label{subsec:uer-external-face}

Let $G$ be a 4-graph with a large-angle assignment $\mathcal{A}(G)$. We show how to select a set of cycles of $C$ and, for each of them, a set of rectangles representing it that bound the external face in any \UER drawing of $G$ that realizes $\mathcal{A}(G)$.

We start with the case in which the four corners $c_1,c_2,c_3,c_4$ are given (Item $1$ of \Cref{le:uerr-fixed-angles}).
We show that, in this case, the external cycle $C$ is unique. Let $(u,v)$ be an edge that has been selected to be part of $C$, initially one of the edges incident to $c_1$. If $\deg_G(u)=4$, we reject the instance, by \Cref{pr:external-face}. If $\deg_G(u)=2$, the other edge incident to $u$ is also part of $C$, and we continue from there.
If $\deg_G(u)=3$, the two edges forming the large angle must belong to $C$.
\begin{wrapfigure}[10]{r}{0.35\textwidth}
     \centering
    \includegraphics[width=0.275\textwidth,page=3]{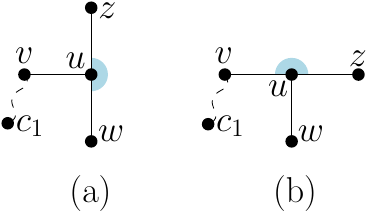}
     \caption{Handling degree-3 vertices on the external face.}
     \label{fig:large-angle}
\end{wrapfigure}
Thus, if $v \notin \mathcal{A}(u)$ (\Cref{fig:large-angle}.a), we reject the instance; otherwise (\Cref{fig:large-angle}.b), the edge creating a large angle around $u$  with $(u,v)$ is in $C$, and we continue from it.
If $c_1$ is reached, we check if the resulting cycle contains $c_2, c_3, c_4$ and if the lengths of the paths connecting them satisfy Condition $(iii)$ of \Cref{pr:external-face}. If so, we can construct a unit edge-length grid rectangle representing~$C$.


If the external cycle is prescribed, but not its corners (Item $2$ of \Cref{le:uerr-fixed-angles}), we check for each degree-3 vertex if its two edges that are on $C$ form the large angle, otherwise we reject the instance. Then, we guess all $O(n^2)$ pairs of adjacent corners, from which we infer the other two corners using Condition $(iii)$ of \Cref{pr:external-face}, and construct the corresponding $O(n^2)$ unit edge-length rectangles.  

If neither the external cycle nor the corners are given (Item $3$ of \Cref{le:uerr-fixed-angles}), note that the external cycle must contain at least one degree-3 vertex, as otherwise $G$ would be a cycle or a non-connected graph. Thus, our strategy is to guess each degree-3 vertex $u$ to be part of the external cycle and use the large-angle assignment to find the corresponding cycle, as in the proof of Item $1$ of \Cref{le:uerr-fixed-angles}. In that proof, we have actually shown that there exists at most one candidate external cycle $C$ that contains $u$. Hence, we do not need to guess any of the other degree-3 vertices that are part of $C$. We claim that a candidate external cycle $C$ contains at least $O(\sqrt{n})$ vertices.
This implies that at the end of the process we have obtained at most $O(\sqrt{n})$ candidate external cycles. Together with Item $2$ of \Cref{le:uerr-fixed-angles}, this implies $O(n^{2.5})$ different unit edge-length grid rectangles. 
We now prove the claim. Given a rectangle $\Gamma_C$, let $h$ and $w$ be the number of vertices different from the corners in the left and in the top side of $\Gamma_C$, respectively. We have that the number of grid points in the interior of $\Gamma_C$ is $i=h\cdot w$, and so $n-|C|\leq i$ must hold. Thus $|C|\geq n-i$, and considering that the value of $i$ is maximized when $h=w$, we have $|C|\geq n-h^2$. Finally, as $|C|=2h+2w+4$ when $h=w$ we obtain $h^2+4h+4-n\geq 0$ that is solved for $h\geq \sqrt{n}-2$. As $|C|>h$, we have that $|C|>\sqrt{n}$.


\subsection{Drawing the internal vertices}\label{subsec:uer-drawing}

From \Cref{subsec:uer-external-face} we obtain either one or $O(n^2)$ or $O(n^{2.5})$ rectangles, respectively, in the three cases of \Cref{le:uerr-fixed-angles}. In the following we show how to place the vertices of $G \setminus C$ in the interior of one of these rectangles in $O(n^2)$ time.

The algorithm is similar to the one in \Cref{th:inner-2-graph}. Namely, we analyze the points of the grid internally to $\Gamma_C$, moving from left to right and secondarily from top to bottom, maintaining two arrays \upA and \leftA. Whenever a point $(i,j)$ is visited, we decide if it is occupied by a vertex (real or dummy) or not (it is internal to a face), based on the values of $\upA[i]$ and $\leftA[j]$, with the same five cases as in \Cref{th:inner-2-graph}. Note that in Case 1 we also check if the placement of vertex $v=\upA[i]=\leftA[j]$ is consistent with its large-angle assignment, if $\deg_G(v)=3$.  

The main difference with respect to \Cref{th:inner-2-graph} is in the update of arrays \upA and \leftA after the placement of $v$, which is easier in this case when $v$ has degree $3$ due to the information deriving from its large-angle assignment. Namely, let $(i,j)$ be the coordinates of $v$. Furthermore, let $u$ and $u'$ be the non-fixed neighbors of $v$, possibly $u'=\nullA$. 

If $\deg_G(v)=2$, we proceed as before: if $v=\leftA[j]$ we set $\leftA[j]=u$ and $\upA[i]=\nullA$, otherwise, we have that $v=\upA[i]$ and we set $\upA[i]=u$ and $\leftA[j]=\nullA$. 

If $\deg_G(v)=3$, then let $\mathcal{A}(v)=\langle w,z\rangle$. We observe that exactly one of $w$ and $z$ is fixed, as pairs formed by the at most two fixed or by the at most two non-fixed neighbors of $v$ cannot create a $180^\circ$ angle. Therefore, we have that  $u=w$, possibly after renaming, and we set $\upA[i]$ or $\leftA[j]$ as $u$, so that $u$ and $z$ are aligned either vertically or horizontally, and the other will be set as $u'$. 

If $\deg_G(v)=4$, then we have no clear way of directly say which of its two non-fixed neighbors should be vertically aligned with $v$, so we will exploit a traversal of the graph as in the algorithm of \Cref{th:inner-2-graph}, starting from $v$ and moving to $u$. Let $w$ be a vertex encountered in this traversal. If $\deg_G(w)=2$, we skip $w$ and continue with its non-visited neighbor. If $\deg_G(w)=4$, we observe that $w$ can be vertically aligned with $v$ if and only if it has exactly one fixed neighbor. Thus we can decide whether $\leftA[j]=u$ and $\upA[i]=u'$ or vice versa by looking at the three neighbors of $w$, or reject as in \Cref{th:inner-2-graph}. If $\deg_G(w)=3$, we first check whether $w \in C$ and, if so, we either reject the instance or update \leftA and \upA accordingly. Otherwise, we distinguish two cases based on whether the last visited vertex $w'$ belongs to $\mathcal{A}(w)$ or not, see \Cref{fig:fpt-traversal}. 

If $w' \notin \mathcal{A}(w)$, we again have that $w$ can be vertically aligned with $v$ if and only if it has exactly one fixed neighbor, which allows us to make a decision as in the degree-4 case. Refer to \Cref{fig:fpt-traversal-down-down} for the case in which it has one fixed neighbor and so it can be vertically aligned, and \Cref{fig:fpt-traversal-right-right} for the case in which it has no fixed neighbor and so it will be horizontally aligned; if $w'$ has more than one fixed neighbor, we reject $\Gamma_C$. 

If $w' \in \mathcal{A}(w)$, then again if $w$ has a fixed neighbor or a neighbor on $C$, we can decide whether $u$ is vertically or horizontally aligned with $v$, see \Cref{fig:fpt-traversal-down-left}. On the other hand, in this case, if $w$ has no placed neighbor we can neither reject the instance nor make a decision, as both assignments are still possible; see \Cref{fig:fpt-traversal-down-right,fig:fpt-traversal-right down,fig:fpt-traversal-up}. In this case, we  treat $w$ as a degree-2 vertex, in the sense that we postpone the decision and move on to the edge incident to $w$ that creates a large angle together with $(w,w')$.
This concludes the description of the algorithm.

\begin{figure}[tb]
        \centering
        \begin{subfigure}[b]{0.32\textwidth}
            \includegraphics[width=\textwidth,page=1]{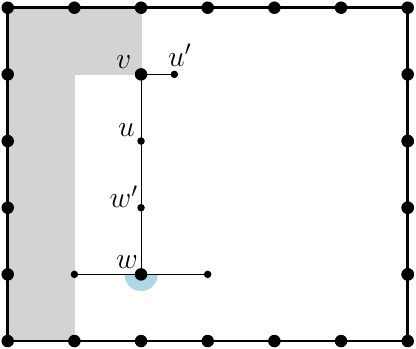}
            \subcaption{}
            \label{fig:fpt-traversal-down-down}
        \end{subfigure}
        \hfil
        \begin{subfigure}{0.32\textwidth}
            \includegraphics[width=\textwidth,page=2]{./figures/fpt-traversal.pdf}
            \subcaption{}
            \label{fig:fpt-traversal-down-left}
        \end{subfigure}
        \hfil
        \begin{subfigure}{0.32\textwidth}
            \includegraphics[width=\textwidth,page=3]{./figures/fpt-traversal.pdf}
            \subcaption{}
            \label{fig:fpt-traversal-down-right}
        \end{subfigure}
        \begin{subfigure}[b]{0.32\textwidth}
            \includegraphics[width=\textwidth,page=4]{./figures/fpt-traversal.pdf}
            \subcaption{}
            \label{fig:fpt-traversal-right-right}
        \end{subfigure}
        \hfil
        \begin{subfigure}{0.32\textwidth}
            \includegraphics[width=\textwidth,page=5]{./figures/fpt-traversal.pdf}
            \subcaption{}
            \label{fig:fpt-traversal-right down}
        \end{subfigure}
        \hfil
        \begin{subfigure}{0.32\textwidth}
            \includegraphics[width=\textwidth,page=6]{./figures/fpt-traversal.pdf}
            \subcaption{}
            \label{fig:fpt-traversal-up}
        \end{subfigure}
        \caption{If $w'\notin \mathcal{A}(w)$, then if $w$ has a fixed neighbor $u$ is vertically aligned with $v$ (a), otherwise it is horizontally aligned (d). If $w'\in \mathcal{A}(w)$, if it has a fixed neighbor, then it is vertically aligned with $v$ (b). In the other cases (c)(e)(f), we continue the traversal along the large angle of $w$ to find another vertex with degree 3 or 4.}
        \label{fig:fpt-traversal}
    \end{figure}

The correctness and the time complexity can be proved similarly to \Cref{th:inner-2-graph}. The decisions based on large-angle assignments have been justified during the description of the algorithm and can be performed in constant time per vertex. It is important to note that each vertex can be visited in a constant number of traversals, as for each of them we either make a final decision or we find a unique edge to follow in the next step of the traversal. The $O(n^2)$ time complexity thus descends from the size of the grid.



\section{Open Problems}\label{sec:conclusions}

We conclude with open problems that arise from our results. 
The complexity of the problem in the general case remains open, as well as the question of whether one can improve our super-linear time algorithms.
%
%
Additionally, it would be interesting to adapt our drawing convention to vertices of high degree (larger than 4). Finally, an extension to 3D drawings could be considered, requiring that the internal regions have a parallelepipedal shape.

\bibliography{biblio-new}

\end{document}